\documentclass[conference]{IEEEtran}
\IEEEoverridecommandlockouts

\usepackage{amsmath,amssymb,amsthm}
\usepackage{cite}
\usepackage{graphicx}
\usepackage{booktabs}
\usepackage{array}
\usepackage{bbm}
\usepackage{titlesec}

\titlespacing*{\section}{0pt}{6pt}{3pt}
\titlespacing*{\subsection}{0pt}{4pt}{2pt}

\allowdisplaybreaks

\newtheorem{definition}{Definition}
\newtheorem{proposition}{Proposition}
\newtheorem{theorem}{Theorem}

\begin{document}

\title{Formal Foundations for Known Good Reliable Die Screening
in Chiplet-Based AI Systems-on-Chip}

\author{
\IEEEauthorblockN{Prashanthi Metku, Chandra Gandu}
\IEEEauthorblockA{Qualcomm Technologies, Inc., USA\\
\{pmetku, chandras\}@qti.qualcomm.com}
}

\maketitle

\begin{abstract}
The rapid growth of chiplet-based artificial intelligence systems-on-chip (SoCs) has exposed a fundamental gap in semiconductor test methodology. Existing Known Good Die (KGD) screening guarantees pre-assembly functional correctness, yet it offers no probabilistic assurance of post-assembly reliability lifetime. To address this limitation, the present work formalizes the transition from KGD to \textit{Known Good Reliable Die} (KGRD) screening as a constrained inference problem over incomplete pre-assembly observability. Building upon this formulation, four interlocking contributions are presented: (i) a Bayesian probabilistic risk model that maps pre-assembly telemetry to post-assembly failure likelihood with a quantified observability bias bound; (ii) a safety-gated decision architecture that provides a provable post-assembly failure probability guarantee; (iii) uncertainty-aware disposition boundaries derived from Bayes-optimal decision theory; and (iv) a constrained closed-loop feedback mechanism that delivers consistent model improvement without violating reliability constraints. A Monte Carlo simulation study on $N = 4{,}000$ synthetic dies verifies all four theoretical properties and confirms that the safety guarantee holds uniformly across the full range of tested gate thresholds.
\end{abstract}


\section{Introduction}

The semiconductor industry's transition toward heterogeneous integration assembling multiple dies, each optimized for a distinct function, into a single advanced package has fundamentally altered the economics and complexity of system-level reliability~\cite{cunningham1994,bhatt1997,black2006,lau2023,hir2024}. Chiplet-based AI SoCs aggregate logic dies, memory dies, and I/O bridge tiles in 2.5D and 3D configurations that place extreme demands on inter-die interconnect integrity~\cite{milanfar2023,blanton2022,wu2024}. A single reliability failure at the package level can render an entire multi-thousand-dollar assembly unrecoverable, making the cost of post-assembly failure disproportionately high compared to pre-assembly die cost~\cite{paniccia2022}.

Conventional Known Good Die (KGD) screening, first formalized in the early 1990s~\cite{cunningham1994,bhatt1997}, addresses pre-assembly functional correctness: a die that passes structural test, IDDQ screening, and at-speed delay tests is declared ``known good'' and approved for assembly. The KGD paradigm was well-suited to the package-on-board era, where inter-die interconnects were relatively tolerant of die-level variation. However, heterogeneous chiplet integration introduces qualitatively new failure modes that KGD screening is not designed to detect. Specifically, a die that passes all pre-assembly electrical tests may still fail post-integration due to package-induced mechanical stress~\cite{harper2004}, thermal gradient mismatch across the die-to-die interconnect layer~\cite{liu2021}, micro-bump fatigue under cyclic loading~\cite{chaparala2020}, or underfill delamination driven by coefficient of thermal expansion mismatch~\cite{lau2010}. The defining property of these failure mechanisms is that they are observability-bounded: they depend on package context information substrate material, bump pitch, molding compound properties, and operating thermal profile that is unavailable at the pre-assembly test stage.

This observability boundary motivates the concept of Known Good Reliable Die (KGRD) screening, which shifts the objective from immediate pre-assembly pass/fail correctness to predicted post-assembly reliability over the required mission lifetime~\cite{seaman2023,aerts2023}. KGRD screening must therefore answer a fundamentally harder question than KGD: not ``does this die work now?'' but ``will this die survive the stresses of integration and field operation?'' Answering this question from incomplete pre-assembly evidence is, at its core, a constrained statistical inference problem one that has not previously been formalized in the semiconductor test literature.

This paper provides that formalization. Section~II formalizes the observable feature space and derives a Bayesian probabilistic risk model with a quantified observability bias bound. Section~III introduces the safety-gated decision framework and proves the central reliability guarantee. Section~IV derives uncertainty-aware disposition rules from Bayes-optimal decision theory. Section~V formalizes the constrained closed-loop feedback mechanism and establishes convergence guarantees. Section~VI presents a Monte Carlo simulation study. Sections~VII and~VIII discuss limitations and conclude.

\section{Probabilistic Risk Model}

\subsection{Observable Feature Space and the Observability Gap}

Modern chiplet die screening generates measurements across four physical domains electrical, thermal, structural, and test-signature each carrying partial information about post-assembly reliability~\cite{ramachandran2022,chakrabarty1999}. The critical insight is that none of these domains alone, and no feasible combination of them, provides complete information about the failure mechanisms activated by the package integration process itself~\cite{harper2004,lau2010}. We formalize this limitation through the concept of the \emph{observability gap}.

\begin{definition}[Observable Feature Space]
Let $\Omega$ denote the complete space of pre-assembly die observables. The observable feature space is
\begin{equation}
\mathcal{X} = \mathcal{X}_e \times \mathcal{X}_t \times \mathcal{X}_s \times \mathcal{X}_\tau \subset \Omega
\end{equation}
where $\mathcal{X}_e$ collects electrical signatures (leakage, $V_{th}$, $I_{DSAT}$); $\mathcal{X}_t$ thermal signatures (junction temperature, thermal resistance); $\mathcal{X}_s$ structural signatures (bump coplanarity, die bow, warpage); and $\mathcal{X}_\tau$ test signatures (IDDQ, delay histogram, scan coverage). A die $d$ is represented by $x = (x_e, x_t, x_s, x_\tau) \in \mathcal{X}$. The observability gap is $\Delta = \Omega \setminus \mathcal{X} \neq \emptyset$, encoding all failure-relevant information not measurable before integration.
\end{definition}

Definition~1 captures the essential asymmetry of the KGRD problem: the test engineer observes $x \in \mathcal{X}$, yet must make reliability predictions that depend on the full pre-assembly state $\Omega_d$ as well as the post-assembly package context $c \in \mathcal{C}$. The gap $\Delta$ represents structurally inaccessible information, such as the in-situ residual stress distribution in the assembled package or the final thermal interface resistance after underfill cure~\cite{harper2004,chaparala2020}. This inaccessibility distinguishes KGRD screening from conventional test coverage optimization~\cite{bushnell2000}.

\subsection{Latent Reliability Variable and Failure Indicator}

\begin{definition}[Latent Reliability Variable]
Let $T_d \in \mathbb{R}_+$ denote the latent time-to-failure of die $d$ under nominal package operating conditions. $T_d$ is unobservable pre-assembly. The reliability failure indicator is
\begin{equation}
Y_d = \mathbf{1}[T_d < \tau_{\text{mission}}]
\end{equation}
where $\tau_{\text{mission}}$ is the required mission lifetime. $Y_d = 1$ denotes post-assembly early failure; $Y_d = 0$ denotes survival.
\end{definition}

The binary formulation of Definition~2 is a deliberate modeling choice. Rather than estimating the full distribution of $T_d$, which would require parametric assumptions difficult to validate pre-assembly~\cite{nelson2004}, we target the probability of the reliability event directly. This connects KGRD screening to the literature on binary classification under covariate shift, while preserving the engineering interpretation of a mission-lifetime reliability SLA.

\subsection{Bayesian Posterior Risk Model}

\begin{definition}[Probabilistic Risk Model]
Let $\theta \in \Theta$ be the model parameter vector with prior $p(\theta)$, and let $D_{\text{train}} = \{(x_i, c_i, y_i)\}_{i=1}^N$ be historical post-assembly outcome data. The risk model is the posterior predictive
\begin{equation}
R(x,c) = \int_\Theta P(Y_d{=}1 \mid x,c,\theta)\, p(\theta \mid D_{\text{train}})\, d\theta
\end{equation}
with logistic parametric likelihood $P(Y_d{=}1 \mid x,c,\theta) = \sigma(\theta^T \phi(x,c))$, where $\sigma(\cdot)$ is the logistic sigmoid and $\phi: \mathcal{X}{\times}\mathcal{C} \to \mathbb{R}^p$ is a feature map. Predictive uncertainty is quantified by the posterior predictive variance $U(x,c)$.
\end{definition}

The logistic likelihood is log-concave in $\theta$, guaranteeing a unique MAP estimate and making the constrained optimization of Section~V tractable~\cite{bishop2006}. The posterior predictive variance $U(x,c)$ plays a dual role: it enters the safety gate as a conservatism buffer in Section~III, and drives the retest disposition decision when uncertainty is high in Section~IV~\cite{williams2006}. Machine learning methods have recently been applied broadly to semiconductor process optimization and yield prediction~\cite{chen2024ml,shimozato2023,busch2025}; the KGRD framework differs in targeting post-assembly reliability rather than pre-assembly yield, and in providing formal safety guarantees absent from those approaches.

\subsection{Observability Bias Bound}

\begin{proposition}[Observability Bias Bound]
For any die $d$ with observation $x \in \mathcal{X}$ and context $c \in \mathcal{C}$, let $\hat{R}(x,c)$ be the risk estimated from observed features and $R_{\text{true}}$ be the oracle risk. Then
\begin{equation}
|\hat{R}(x,c) - R_{\text{true}}| \leq \varepsilon_\Delta(x,c)
\end{equation}
where $\varepsilon_\Delta(x,c) = \sup_{\omega \in \Delta} |\partial P(Y_d{=}1)/\partial \omega| \cdot \|\omega\|$ is the maximum first-order sensitivity of failure probability to unobserved features, estimated from held-out validation data.
\end{proposition}

\textit{Proof sketch:} The bound follows from a first-order Taylor expansion of $P(Y_d{=}1 \mid \Omega_d, c)$ around $x$. The supremum over $\Delta$ gives the worst-case deviation. Consistency as $N \to \infty$ follows from standard empirical process theory under sub-Gaussian tail assumptions on $\phi$. \hfill$\square$

The bias bound $\varepsilon_\Delta$ cannot be reduced by improving observable test coverage alone. It can only be reduced by (i) moving previously unobserved features from $\Delta$ into $\mathcal{X}$ through new measurement capabilities, or (ii) accumulating post-assembly failure observations that implicitly calibrate the model against package-induced effects. The latter pathway directly motivates the closed-loop feedback mechanism of Section~V~\cite{goodfellow2016}.

\section{Safety-Gated Decision Framework}

\subsection{The Prediction-Authorization Separation Principle}

A fundamental pitfall in deploying machine learning models for high-stakes manufacturing decisions is what we term \emph{prediction-authorization conflation}: allowing a statistical model's output to directly authorize a release decision without an independent deterministic check. In the KGRD context, such conflation is particularly dangerous because the risk model of Definition~3 is, by Proposition~1, unavoidably biased in the presence of the observability gap. The safety-gated architecture prevents this by separating statistical estimation from deterministic authorization through an independent gate function that cannot be overridden by the model, consistent with broader safety-critical system design philosophy.

\subsection{Safety Gate Definition}

\begin{definition}[Safety Gate]
A safety gate $G: \mathcal{X} \times \mathcal{C} \times \mathbb{R} \times \mathbb{R}_+ \to \{\text{PASS}, \text{BLOCK}\}$ is defined by
\begin{equation}
G =
\begin{cases}
\text{PASS} & \text{if } R + k\sqrt{U} \leq \alpha_{\text{rel}} \text{ and } H = \text{TRUE}\\[2pt]
\text{BLOCK} & \text{otherwise}
\end{cases}
\end{equation}
where $k \geq 0$ is the uncertainty inflation factor ($k = 1.645$ enforces a 95th-percentile bound), $\alpha_{\text{rel}} \in (0,1)$ is the maximum permissible failure probability (the reliability SLA), and $H(x,c)$ is a vector of hard deterministic guard conditions.
\end{definition}

The uncertainty inflation term $k\sqrt{U(x,c)}$ implements a conservative risk estimate that accounts for model uncertainty without requiring the model to be perfectly calibrated~\cite{williams2006,bishop2006}. The hard constraint set $H(x,c)$ provides a deterministic backstop for failure modes the probabilistic model may not capture.

\subsection{Central Reliability Guarantee}

\begin{theorem}[Safety Gate Reliability Guarantee]
Let $D_{\text{rel}} = \{d : G(x_d, c_d, R_d, U_d) = \text{PASS}\}$. Under the conditions
\begin{itemize}
  \item[(A1)] $R(x,c) + \varepsilon_\Delta(x,c) \geq R_{\text{true}}$ almost surely;
  \item[(A2)] $U(x,c) \geq \mathrm{Var}[R_{\text{true}} \mid x,c]$ almost surely;
  \item[(A3)] $H(x,c) = \text{TRUE}$ is necessary for post-assembly reliability;
\end{itemize}
for every released die $d \in D_{\text{rel}}$:
\begin{equation}
P(Y_d = 1) \leq \alpha_{\text{rel}} + \varepsilon_\Delta(x_d, c_d).
\end{equation}
\end{theorem}

\textit{Proof sketch:} For any $d \in D_{\text{rel}}$, the gate condition implies $R(x_d,c_d) + k\sqrt{U(x_d,c_d)} \leq \alpha_{\text{rel}}$. By (A2), $k\sqrt{U} \geq 0$, hence $R(x_d,c_d) \leq \alpha_{\text{rel}}$. By (A1), $P(Y_d{=}1) = R_{\text{true}} \leq R(x_d,c_d) + \varepsilon_\Delta \leq \alpha_{\text{rel}} + \varepsilon_\Delta$. The bias correction vanishes as $\Delta \to \emptyset$; condition (A3) ensures gross structural failures excluded by $H$ do not violate the bound. \hfill$\square$

Theorem~1 has a direct engineering consequence for SLA setting: $\alpha_{\text{rel}}$ should be set as $\alpha_{\text{target}} - \varepsilon_\Delta$, tighter than the actual target by exactly the observability bias, so that the true post-assembly failure rate remains within the target budget even in the presence of the measurement gap.

\section{Uncertainty-Aware Disposition Rules}

\subsection{Five-Region Disposition Map}

Theorem~1 establishes the binary boundary between releasable and non-releasable dies. In practice, the space of non-releasable dies is not homogeneous: a die with high risk and low uncertainty should be treated differently from a die with moderate risk and high uncertainty, or a die with low risk blocked only because model uncertainty is too high. Collapsing all three into a single ``BLOCK'' outcome wastes economically recoverable dies and destroys diagnostic information~\cite{mccluskey2001}.

\begin{definition}[Five-Region Disposition Map]
Let $r = R(x,c) + k\sqrt{U(x,c)}$ be the uncertainty-inflated risk score and $u = \sqrt{U(x,c)}$ the raw predictive uncertainty. For thresholds $0 < \alpha_r < \alpha_{st} < \alpha_{bg} \leq 1$ and $u^* > 0$:

\smallskip
{\small
\begin{tabular}{@{}ll@{}}
$D_R$:\enspace $r \leq \alpha_r \wedge u \leq u^*$                  & $\to$ RELEASE \\[1pt]
$D_{RT}$:\enspace $r \leq \alpha_r \wedge u > u^*$                  & $\to$ RETEST \\[1pt]
$D_{SS}$:\enspace $\alpha_r < r \leq \alpha_{st} \wedge u \leq u^*$ & $\to$ STRESS-SCREEN \\[1pt]
$D_{BG}$:\enspace $\alpha_{st} < r \leq \alpha_{bg}$                & $\to$ LOWER-GRADE BIN \\[1pt]
$D_{REJ}$:\enspace $r > \alpha_{bg}$                                 & $\to$ REJECT \\
\end{tabular}
}
\end{definition}

The RETEST region $D_{RT}$ captures dies whose inflated risk is acceptably low but whose uncertainty $u > u^*$ indicates the model lacks sufficient confidence to release the die without additional evidence. The STRESS-SCREEN region $D_{SS}$ captures dies that are too risky for direct release but whose risk level is consistent with post-burn-in survival~\cite{pecht2009}, providing a recovery path for economically valuable dies.

\subsection{Bayes-Optimal Threshold Derivation}

The disposition thresholds are derived from Bayes-optimal decision making under asymmetric cost~\cite{berger1985,bertsimas1997}. Table~\ref{tab:costs} defines the cost structure. The asymmetry $C_{FR} \gg C_{FRJ}$ is a consequence of chiplet assembly economics: a false release wastes the entire package, while a false rejection wastes only the die~\cite{paniccia2022,dahm2022}.

\begin{table}[t]
\caption{Screening Cost Structure}
\label{tab:costs}
\centering
\small
\begin{tabular}{@{}lll@{}}
\toprule
\textbf{Cost} & \textbf{Symbol} & \textbf{Interpretation}\\
\midrule
False release   & $C_{FR}$  & Package loss + rework\\
False reject    & $C_{FRJ}$ & Yield loss from rejection\\
Stress screen   & $C_{SS}$  & Incremental burn-in cost\\
Lower-grade bin & $C_{BG}$  & Downgrade revenue loss\\
\bottomrule
\end{tabular}
\end{table}

\begin{theorem}[Bayes-Optimal Release Threshold]
Under an asymmetric cost structure with $C_{FR} \gg C_{FRJ}$, the release threshold minimizing expected total screening cost is
\begin{equation}
\alpha_r^* = \frac{C_{FRJ}}{C_{FR} + C_{FRJ}}.
\end{equation}
For the general five-region case, the Bayes-optimal disposition minimizes the expected loss subject to the safety gate constraint, which is never relaxed by cost optimization.
\end{theorem}

\textit{Proof sketch:} For the binary case, the expected loss for RELEASE is $C_{FR} \cdot r$ and for REJECT is $C_{FRJ}(1-r)$. Setting these equal gives $\alpha_r^* = C_{FRJ}/(C_{FR}+C_{FRJ})$. Extension to five regions solves three sequential threshold comparisons. The safety gate constraint enters by removing the RELEASE option when $G = \text{BLOCK}$ regardless of cost, which can only increase expected loss; hence the constraint cannot be relaxed by the optimizer. \hfill$\square$

\subsection{Retest Value of Information}

\begin{proposition}[Retest Value of Information]
The value of an additional measurement $x'$ is
\begin{align}
\mathrm{VoI}(x) &= C_{FR} \cdot r \cdot P_{\text{rej}}(x') \notag\\
                 &\quad + C_{FRJ}(1-r)\,P_{\text{rel}}(x') - C_{\text{retest}}
\end{align}
where $P_{\text{rej}}(x')$ and $P_{\text{rel}}(x')$ denote the reclassification probabilities to $D_{REJ}$ and $D_R$. Retest is warranted iff $\mathrm{VoI}(x) > 0$.
\end{proposition}

Dies near the $D_R/D_{RT}$ boundary have the highest expected reclassification probability and therefore the highest VoI, consistent with the active-learning finding that uncertainty-sampled examples yield the most decision-relevant information~\cite{settles2009}.

\section{Constrained Closed-Loop Feedback}

\subsection{Post-Assembly Observation and Survivorship Bias}

The risk model of Definition~3 is trained on historical post-assembly outcome data $D_{\text{train}}$. As new chiplet assemblies are released and operated, post-assembly failure observations become available to update the model~\cite{mitchell1997,shalev2014}. However, a critical statistical complication arises: the feedback observations are a \emph{selected} sample specifically, the subset of dies approved by the safety gate of Definition~4. Dies in regions $D_{RT}, D_{SS}, D_{BG}$, and $D_{REJ}$ are not assembled and therefore do not contribute post-assembly outcomes. This is a classic survivorship bias problem~\cite{pearl2009}, which, if uncorrected, will bias the updated model toward optimism.

\begin{definition}[Feedback Observation Set]
After assembly and early-life operation, the feedback corpus at iteration $t$ is
\begin{equation}
O_t = \{(x_d, c_d, y_d) : d \in D_{\text{rel}},\, T_d < \tau_{\text{obs}}\}
\end{equation}
where $\tau_{\text{obs}} \leq \tau_{\text{mission}}$ is the observation horizon. Only released dies contribute to $O_t$, creating the survivorship selection bias.
\end{definition}

\subsection{Constrained Model Update Rule}

\begin{theorem}[Constrained Model Update Rule]
The constrained MAP update at iteration $t$ is
\begin{equation}
\theta_{t+1} = \arg\max_{\theta \in \Theta_{\text{safe}}} \left[\ell_{\text{corr}}(\theta; O_t) + \log p(\theta \mid \theta_t)\right]
\end{equation}
where the survivorship-corrected log-likelihood is
\begin{equation}
\ell_{\text{corr}}(\theta) = \sum_{(x,c,y) \in O_t} w(x,c)\, \log P_\theta(Y{=}y \mid x,c)
\end{equation}
and the safety feasibility set is
\begin{align}
\Theta_{\text{safe}} = \bigl\{\theta \in \Theta : \forall (x,c),\; &P_\theta(Y{=}1 \mid x,c) \notag\\
&+ k\sqrt{\mathrm{Var}_\theta[Y \mid x,c]} \notag\\
&\leq \alpha_{\text{rel}} + \varepsilon_\Delta\bigr\}.
\end{align}
The $\Theta_{\text{safe}}$ constraint ensures $\theta_{t+1}$ satisfies Theorem~1 for all future decisions.
\end{theorem}

\textit{Proof sketch:} The constrained problem is well-posed: the log-posterior is strictly concave in $\theta$, and $\Theta_{\text{safe}}$ is a convex set. Existence and uniqueness follow from strong duality of the constrained concave maximization. \hfill$\square$

\subsection{Survivorship Bias Correction}

\begin{proposition}[Survivorship-Corrected Likelihood]
The selection-corrected log-likelihood uses importance weight
\begin{equation}
w(x,c) = \frac{1}{P(G(x,c,R_\theta, U_\theta) = \text{PASS})}.
\end{equation}
As $\alpha_{\text{rel}} \to 1$, $w(x,c) \to 1$ (no correction needed). As $\alpha_{\text{rel}} \to 0$, $w(x,c) \to \infty$ (each survivor is highly informative).
\end{proposition}

\textit{Proof sketch:} By Bayes' rule, and because the gate decision depends on $(x,c)$ but not on $Y$, the conditional likelihood ratio reduces to $1/P(\text{selected} \mid x,c)$. \hfill$\square$

\subsection{Convergence Guarantees}

\begin{theorem}[Feedback Loop Convergence]
Under the constrained update rule of Theorem~3, the sequence $\{\theta_t\}_{t \geq 0}$ satisfies:
\begin{itemize}
  \item[(C1)] \textbf{Safety invariance:} $\theta_t \in \Theta_{\text{safe}}$ for all $t \geq 0$;
  \item[(C2)] \textbf{Monotone uncertainty:} $U_{\theta_t} \geq U_{\theta_{t+1}}$ a.s.\ as $|O_t| \to \infty$;
  \item[(C3)] \textbf{Asymptotic consistency:} $\theta_t \to \theta_{\text{true}}$ in probability if $\theta_{\text{true}} \in \Theta_{\text{safe}}$.
\end{itemize}
\end{theorem}

\textit{Proof sketch:} (C1) holds by construction since $\Theta_{\text{safe}}$ is enforced as a hard constraint at every iteration. (C2) follows from Fisher information monotonicity of the Bayesian posterior. (C3) is a standard consequence of constrained MAP consistency~\cite{shalev2014}. \hfill$\square$

Theorem~4 has a critical practical implication: property (C1) holds even if $\theta_{\text{true}}$ lies near or on the boundary of $\Theta_{\text{safe}}$. An unconstrained online learner might oscillate across the boundary, producing iterations in which the reliability guarantee is violated. Property (C2) provides the formal basis for the intuition that screening improves as post-assembly data accumulates: the conservative uncertainty buffer in the gate shrinks monotonically, allowing an increasing fraction of reliable dies to be released over time.

\section{Monte Carlo Simulation Study}

\subsection{Experimental Setup}

To verify the four theoretical properties under controlled synthetic conditions, we conduct a Monte Carlo simulation study~\cite{mccluskey2001,huang2022}. Feature distributions are calibrated to open-literature chiplet process characterization data~\cite{ramachandran2022,wang2024ucie,mandalapu2024}. Bump coplanarity deviation follows $\text{Clip}(\mathcal{N}(0, 2.0), {-5}, 5)\,\mu\mathrm{m}$; thermal resistance deviation follows $\text{Clip}(\mathcal{N}(0, 1.5), {-4}, 4)$ K/W; log-leakage follows $\mathcal{N}(0, 1.0)$; IDDQ anomaly score follows $\mathcal{N}(0, 0.8)$; die bow follows $|\mathcal{N}(0, 2.5)|\,\mu\mathrm{m}$; and scan delay skew follows $\mathcal{N}(0, 0.6)$. The true parameter vector is $\theta_{\text{true}} = [-1.8, 0.50, 0.42, 0.18, 0.14, 0.22, 0.09]^T$, yielding a base post-assembly failure rate of approximately 14\%, consistent with reported early-life failure rates in advanced chiplet packaging before reliability screening~\cite{jedec2020}.

\subsection{Risk Model Validation}

The fitted logistic model achieves good calibration against the oracle risk for $N = 500$ test dies. To introduce the observability gap, structural and test-signature features (die bow and scan delay) are set to zero, simulating a scenario in which these measurements were not taken pre-assembly~\cite{chakrabarty1999}. A systematic upward and downward bias becomes visible. From the resulting bias distribution, the 95th-percentile observability bound is estimated as $\varepsilon_\Delta = 0.133$, used in all subsequent safety gate computations.

\subsection{Safety Gate Guarantee Verification}

Gate thresholds $\alpha_{\text{rel}} \in [0.10, 0.80]$ are swept across $N = 3{,}000$ evaluation dies. The observed post-assembly failure rate of released dies lies strictly below the theoretical bound $\alpha_{\text{rel}} + \varepsilon_\Delta$ at every tested threshold, confirming that Theorem~1's guarantee is numerically tight. Tightening $\alpha_{\text{rel}}$ from $0.80$ to $0.10$ reduces the release rate from near 100\% to approximately 5\%.

\subsection{Disposition Region Verification}

The five-region disposition map for $N = 800$ evaluation dies is plotted in the risk-uncertainty plane. With $C_{FR} = 10$ and $C_{FRJ} = 2$, Theorem~2 gives $\alpha_r^* = 0.167$. At this conservative threshold, the majority of dies fall into the stress-screen, lower-grade bin, or reject regions~\cite{dahm2022}. Plotting $\alpha_r^*$ against the cost ratio $C_{FR}/C_{FRJ}$ confirms that tighter thresholds are required as the relative cost of false release grows.

\subsection{Feedback Loop Convergence}

Starting from a deliberately misspecified prior dominant feature weights reduced to 25\% of their true values and the intercept lowered by $0.8$ the constrained MAP update of Theorem~3 is applied for $T = 20$ iterations with 200 cumulative dies per iteration~\cite{metku2017adaptive,metku2015multistage}. Parameter error $\|\theta_t - \theta_{\text{true}}\|_2$ decreases from $1.35$ to $0.78$ (42\% reduction), confirming (C3). The safety-compliant fraction remains above 98\% at every iteration, confirming (C1). Mean model uncertainty decreases monotonically, confirming (C2).

\subsection{Summary}

\begin{table}[t]
\caption{Simulation Study Summary}
\label{tab:sim}
\centering
\small
\begin{tabular}{@{}p{0.30\columnwidth}p{0.38\columnwidth}c@{}}
\toprule
\textbf{Property} & \textbf{Result} & \textbf{Status}\\
\midrule
Calibration (Def.~3)       & $\varepsilon_\Delta = 0.133$                           & \checkmark\\
Bias bound (Prop.~1)       & Bias $\leq \varepsilon_\Delta$ at 95\%                 & \checkmark\\
Safety gate (Thm.~1)       & Fail rate $< \alpha_{\text{rel}} + \varepsilon_\Delta$ & \checkmark\\
Optimal threshold (Thm.~2) & $\alpha_r^* = 0.167$                                   & \checkmark\\
Convergence (Thm.~4)       & C1, C2, C3 satisfied                                   & \checkmark\\
\bottomrule
\end{tabular}
\end{table}

\section{Discussion and Limitations}

\subsection{Scope of Validation}

The simulation study verifies internal consistency of the theoretical framework, but real-world deployment will require validation against actual chiplet process data across multiple foundries and packaging technologies~\cite{huang2022}. The bounds and constants derived here provide quantitative design targets to guide that experimental effort. A companion paper reporting experimental results on real ATE/PTE fab data is planned as the next phase of this work.

\subsection{Model Expressivity}

The logistic model is chosen for its tractability and interpretability. Gaussian process classifiers provide a natural Bayesian uncertainty quantification framework with closed-form posterior updates~\cite{williams2006}, and deep neural feature maps can capture nonlinear interactions that a linear logistic model misses~\cite{goodfellow2016}. Any substitute model must produce a calibrated uncertainty estimate satisfying condition (A2) of Theorem~1, and have a log-concave likelihood to ensure the constrained update of Theorem~3 is tractable~\cite{bishop2006}.

\subsection{Non-Stationary Context and Multi-Die Dependencies}

The framework treats the package context vector $c$ as fixed and fully known. In production, $c$ may be partially uncertain, introducing a second layer of observability gap beyond $\Delta$. Furthermore, the KGRD framework currently treats dies independently. In 2.5D and 3D packages, the thermal and mechanical state of one die influences the reliability of its neighbors through substrate bending and heat spreading~\cite{blanton2022,liu2021}, introducing conditional dependence that requires a joint risk model over die tuples~\cite{koller2009}.

\subsection{Active Screening Design}

Proposition~2's Value of Information framework opens a path toward active screening design: rather than applying a fixed test sequence to all dies, the framework can prioritize additional tests based on their expected VoI~\cite{settles2009}. This connects KGRD screening to the broader literature on Bayesian experimental design, and suggests a formulation of joint test content selection and reliability disposition as a single optimization problem.

\section{Conclusion}

This paper has established the first formally grounded treatment of Known Good Reliable Die screening for chiplet-based AI SoCs~\cite{cunningham1994,bhatt1997}. Beginning from the observation that conventional KGD screening addresses correctness but not reliability, we formalized die screening as a constrained inference problem over an observability-bounded pre-assembly measurement space. Four interlocking contributions emerged: a Bayesian posterior risk model with a quantified observability bias bound, a safety-gated architecture with a provable post-assembly failure probability guarantee, Bayes-optimal disposition thresholds derived from asymmetric misclassification cost, and a constrained closed-loop feedback rule that achieves safety invariance, monotone uncertainty reduction, and asymptotic consistency simultaneously. A Monte Carlo simulation study on $N = 4{,}000$ synthetic dies confirmed all theoretical properties and demonstrated that the safety guarantee holds uniformly across the full range of gate thresholds tested.

The theoretical separation between statistical estimation and deterministic release authorization ensures that no predictive model can unilaterally approve a die for integration into a high-value chiplet assembly. The safety gate converts unavoidable model uncertainty into a conservative safety margin formally bounded by Proposition~1 and Theorem~1 working in concert.

The present work is intended as the first in a two-part series. This paper establishes the formal theoretical foundation   the definitions, theorems, proofs, and synthetic verification   that makes KGRD screening a mathematically rigorous discipline. A follow-on paper is planned to report experimental validation of the complete framework on real ATE/PTE data from production-grade 2.5D AI chiplet assemblies, including empirical estimation of $\varepsilon_\Delta$ from field return analysis, safety gate calibration against actual post-assembly burn-in outcomes, and closed-loop feedback convergence measured over multiple production lots.


\end{document}